\documentclass[11pt]{article}

\usepackage{amsmath}
\usepackage{amsthm}
\usepackage{amssymb}
\usepackage[pdftex]{color,graphicx}
\usepackage{caption}
\usepackage{subcaption}
\usepackage{setspace}
\usepackage{fullpage}

\usepackage{pgf}
\usepackage{tikz}
\usetikzlibrary{arrows,automata}

\usepackage[all]{xy}
\usepackage{epigraph}
\usepackage{authblk}

\newtheorem{thm}{Theorem}
\newtheorem{cor}[thm]{Corollary}

\newtheorem{prop}[thm]{Proposition}

\theoremstyle{definition}
\newtheorem{myDef}[thm]{Definition}

\title{Complexity of evolutionary equilibria in static fitness landscapes}
\author[1]{Artem Kaznatcheev\thanks{artem.kaznatcheev@mail.mcgill.ca}}
\date{August 22, 2013}
\affil[1]{School of Computer Science, McGill University, Montreal, Canada}

\begin{document}

\maketitle

\abstract{A fitness landscape is a genetic space -- with two genotypes adjacent if they differ in a single locus -- and a fitness function. Evolutionary dynamics produce a flow on this landscape from lower fitness to higher; reaching equilibrium only if a local fitness peak is found. I use computational complexity to question the common assumption that evolution on static fitness landscapes can quickly reach a local fitness peak. I do this by showing that the popular $NK$ model of rugged fitness landscapes is PLS-complete for $K \geq 2$; the reduction from Weighted 2SAT is a bijection on adaptive walks, so there are $NK$ fitness landscapes where every adaptive path from some vertices is of exponential length. Alternatively -- under the standard complexity theoretic assumption that there are problems in PLS not solvable in polynomial time -- this means that there are no evolutionary dynamics (known, or to be discovered, and not necessarily following adaptive paths) that can converge to a local fitness peak on all $NK$ landscapes with $K = 2$. Applying results from the analysis of simplex algorithms, I show that there exist single-peaked landscapes with no reciprocal sign epistasis where the expected length of an adaptive path following strong selection weak mutation dynamics is $e^{O(n^{1/3})}$ even though an adaptive path to the optimum of length less than $n$ is available from every vertex. The technical results are written to be accessible to mathematical biologists without a computer science background, and the biological literature is summarized for the convenience of non-biologists with the aim to open a constructive dialogue between the two disciplines.}

\newpage
\epigraph{Nothing in biology makes sense except in the light of evolution.}{Theodosius Dobzhansky}

\section{Introduction}

At the same time (1936-1947) as the birth of computer science, biologists were building the modern evolutionary synthesis: reconciling Mendelian genetics with Darwin's evolution by natural selection, and explaining the broad-scale changes observed by palaeontologists in terms of changes in local populations and their genes. They unified several branches of biology that had been diverging at the dawn of the 20th century into a single paradigm that persists today. In the intervening years, the importance of evolution has not decreased, and mathematical biology and modeling of evolutionary adaptation only grew. Surprisingly, fundamental questions about the genetic basis of adaptation remained unanswered~\cite{O05}, or even unasked: how quickly does evolution proceed? What is the distribution of effects on fitness during a typical adaptation? Does the magnitude of the effect change as the process approaches an optimum? Is the assumption of reachable local optima reasonable? My biological goal here is to ask the first and last question.

My secondary goal is to turn the lens of theoretical computer science onto the genetic theory of adaptation, and provide answers to these two questions. The most prominent attempt of theoretical computer science to study the evolution and adaptation is Valiant's theory of evolvability~\cite{evolvability}. For Valiant, evolution is a subset of PAC learning~\cite{PAC}, in particular it is equivalent to a restricted version of statistical query learning~\cite{F08}. This approach has generated a number of technical proofs and spawned a small subfield of research. Unfortunately, all of the attention has come from computer scientists, with no biologists working with the model. Although the model is a natural way for computational learning theorists to think about evolution, it is not expressed in languages and conceptual metaphors familiar to biologists. Hence, it is difficult for them to build on this model or use its results. As a newcomer to biology, my method is non-prescriptive; unlike the machine learning approach to evolution, I do not suggest new metaphors or models of adaptation, but offer new tools -- combinatorics, analysis of algorithms, and computational complexity -- to study extant models familiar to mathematical biologists. As such, I focus on presenting original results in a matter accessible to a biologist with mathematical training, but without a necessary focus in the details of theoretical computer science. By making a conscious effort to use the language of biology, and to connect (at least in a qualitative way) to empirical data, I hope to introduce theoretical computer scientists to the models and methods used by biologists. Instead of proving hard technical results, I concentrate on raising interesting questions for both computer scientists and biologists; the aim is to open a constructive dialogue between the two disciplines.

To start the conversation, I focus on a popular conceptual metaphor in biology -- the fitness landscape. I begin with biological background in section~\ref{sec:bio}, with an emphasis on epistasis (section~\ref{sec:epistasis}) and a breif summary of our empirical knowledge (section~\ref{sec:empirical}). I introduce the most popular model of fitness landscapes -- the $NK$ model -- in section~\ref{sec:NK}. In section~\ref{sec:PLS}, I prove that this model is PLS-complete for $K \geq 2$. The reduction I build is from Weighted 2SAT~\cite{SY91} and a bijection on adaptive walks. Thus, there are $NK$ fitness landscapes where every adaptive path from some vertices to a local maximum is of exponential length -- an unconditional result. Alternatively -- under the standard complexity theoretic assumption that there are problems in PLS not solvable in polynomial time -- this means that there are no evolutionary dynamics that can converge to a local fitness peak on all $NK$ landscapes with $K = 2$. The dynamics in this conditional result are unrestricted (except being polynomial time) and can include tracking populations, sexular reproduction, recombination, strong mutation, and neutral drift; although the specification of what it means to be a fitness peak remains fixed. In section~\ref{sec:simplex}, I adapt a theorem from the analysis of simplex algorithms~\cite{MS06} to show that there exist non-rugged landscapes with a single fitness peak and no reciprocal sign epistasis (see section~\ref{sec:epistasis} for a definition), where strong selection weak mutation dynamics take $e^{O(n^{1/3})}$ steps to find the unique peak, even though an adaptive path of length less than $n$ exists from every vertex. I briefly discuss the consequences of these results for biology in the final section.

\section{Biological background}
\label{sec:bio}

\subsection{Fitness landscapes}

In 1932, Wright introduced the metaphor of a fitness landscape~\cite{W32}. The landscape is a genetic space where each vertex is a possible genotype and an edge exists between two vertices if a single mutation transforms the genotype of one vertex into the other. In the case of a biallelic system we have $n$ loci (positions), at each of which it is possible to have one of two alleles, thus our space is the $n$-bit binary strings $\{0,1\}^n$~\footnote{We could also look at spaces over larger alphabets, $4$ letters for sequence space of DNA, or $20$ letters for amino acids; but the biallelic system is sufficiently general for us.}. A mutation can flip any loci from one allele to the other, thus two strings $x,y \in \{0,1\}^n$ are adjacent if they differ in exactly one bit; the landscape is an $n$-dimensional hypercube with genotypes as vertices. The last ingredient, fitness, is given by a function that maps each string to a non-negative real number.

Individual organisms can be thought of as inhabiting the vertices of the landscape corresponding to their genotype. The probability (or rate, for asexual organisms) of each organism reproducing is proportional to the fitness value of its vertex\footnote{Note that this fitness is independent of the distribution of other agents, and hence this model is only valid for frequency-independent selection. This is an overreach of reductionism in ignoring the distribution of the whole population, but one that can sometimes be theoretically justified and at other times is conceptually or empirically convenient. In practice, this makes the model only valid on relatively short timescales since organisms tend to change their environment (and thus selective pressures) over longer timescales.} 

\subsection{Strong-selection weak-mutation, fitness graphs, \& local equilibria}

Wright's original formulation considered a population of sexually reproducing organisms distributed over the landscape. Modern treatments, however, usually concentrate on asexual populations. In an asexual population with sufficiently low mutations, the probability of a double (or more) mutation during reproduction is so low that it can be ignored and only point mutations matter. In this setting, an organism's offspring will be in the same or an adjacent vertex as its parent. In a population with mutation rate $u$ and effective population size $M$, if $M \log M = o(1/u)$ then a beneficial mutation goes to fixation in the population before the next one can occur -- the population remains monomorphic at most times and can be treated as a single point on the fitness landscape~\cite{G83,G84}. This assumption is known as \emph{strong-selection weak-mutation (SSWM)} and allows us to coarse-grain our scales and model evolutionary dynamics as steps from a vertex of lower fitness to an adjacent vertex of higher fitness.\footnote{This hill climbing view of evolution has been very useful for proving mathematical results, but has also unfortunately led to a view of evolution as an optimization procedure and corresponding teleological conclusions.}

\begin{myDef}
	In a fitness landscape with fitness $f$, a path $v_1...v_t$ is called
	\emph{adaptive} if each $v_{i + 1}$ differs from $v_{i}$ by one bit and $f(v_{i + 1}) \geq f(v_{i})$.
\end{myDef}

\begin{myDef}
	\emph{Strong-selection weak-mutation (SSWM) dynamics}\footnote{In the simplex algorithm literature, this definition is equivalent to the \emph{random edge} pivot rule. We will use this observation in section~\ref{sec:simplex}} extends an adaptive path $v_1...v_t$ that has not reached a local fitness maximum, by sampling uniformly at random a vertex adjacent to $v_t$ of higher fitness.
\end{myDef}

In the SSWM model, the exact quantitative fitness differences between adjacent genotypes does not matter, only if the difference is positive or negative. As such, we can replace the fitness function by a flow on the graph: for adjacent vertices, direct the edges from the lower to the higher fitness genotype. Assuming that no two adjacent genotypes have exactly the same fitness this results in a characterization of fitness landscapes as directed acyclic graphs on $\{0,1\}^n$. Crona, Greene, \& Barlow~\cite{CGB13} introduced this representation into theoretical biology as \emph{fitness graphs}, but they have been used previously in empirical studies of fitness landscapes~\cite{dVPK09,FKdVK11,Getal13,SSFKdV13}. This approach is particularly useful empirically because it is difficult to quantitatively compare fitnesses across experiments. In theoretic work, the fitness graph approach has made the proofs of some classical theorems relating local structure to global properties easier.

An adaptive walk will be at evolutionary equilibrium only if it reaches a local or global peak (a sink in the flow model). The fact that (in finite fitness landscapes) it eventually has to reach a peak has been taken by biologists as justification for assuming that the population is already at a peak. In fact, from the earliest days to today, investigations in biology had the form~\cite{W32}:

\begin{quote}
In a rugged field of this character selection will easily carry the species to the nearest peak, but there may be innumerable other peaks which are higher but which are separated by ``valleys." The problem of evolution as I see it is that of a mechanism by which the species may continually find its way from lower to higher peaks in such a field.
\end{quote}

Note that in the above passage, Wright implicitly assumes that a population starting away from equilibrium will get to a peak in a reasonable amount of time.

\subsection{Epistasis}
\label{sec:epistasis}

\emph{Epistasis} is the amount of inter-loci interactions and is usually considered for the interaction of just two loci at a time. For two loci, there are 3 types of epistasis: magnitude, sign, and reciprocal sign. Consider two loci with the first having alleles a and A, and the second b and B. Assume that the upper-case combination is more fit $f(ab) < f(AB)$. If there is \emph{no epistasis} then the fitness effects are additive and independent of background: $f(AB) - f(aB) = f(Ab) - f(ab)$, $f(AB) - f(Ab) = f(aB) - f(ab)$. In \emph{magnitude epistasis} this additivity is broken, but the signs remain: $f(AB) > f(aB) > f(ab)$ and $f(AB) > f(Ab) > f(ab)$.~\footnote{Fitness graphs do not distinguish between no epistasis and magnitude epistasis~\cite{CGB13}, and neither will I when dealing with SSWM dynamics.}

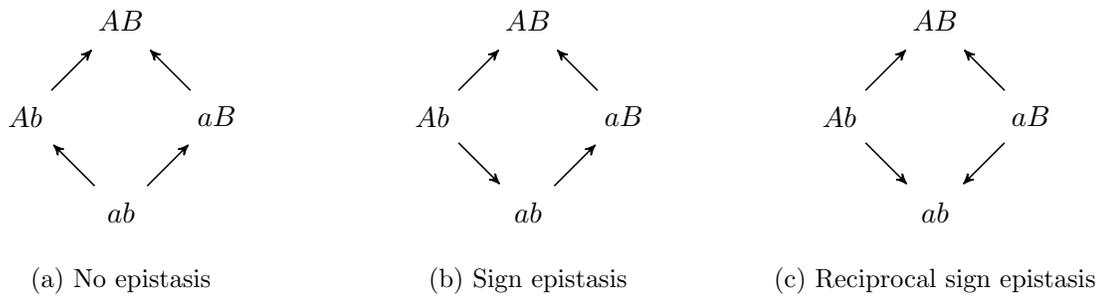
\begin{figure}
\centering
\begin{subfigure}[b]{0.32\textwidth}
\centering
\begin{tikzpicture}[->,>=stealth',shorten >=1pt,auto,node distance=1.8cm,
                    semithick]
  \tikzstyle{every state}=[fill=white,draw=white,text=black]

  \node[state](ab){$ab$};
  \node[state](Ab)[above left of=ab]{$Ab$};
  \node[state](aB)[above right of=ab]{$aB$};
  \node[state](AB)[above right of=Ab]{$AB$};

  \path (ab) edge              node {} (Ab)
            edge node {} (aB)
        (Ab) edge node {} (AB)
        (aB) edge node {} (AB)
            ;
\end{tikzpicture}
\caption{No epistasis}
\end{subfigure} \begin{subfigure}[b]{0.32\textwidth}
\centering
\begin{tikzpicture}[->,>=stealth',shorten >=1pt,auto,node distance=1.8cm,
                    semithick]
  \tikzstyle{every state}=[fill=white,draw=white,text=black]

  \node[state](ab){$ab$};
  \node[state](Ab)[above left of=ab]{$Ab$};
  \node[state](aB)[above right of=ab]{$aB$};
  \node[state](AB)[above right of=Ab]{$AB$};

  \path (ab) edge node {} (aB)
        (Ab) edge node {} (AB)
        	edge node {} (ab)
        (aB) edge node {} (AB)
            ;
\end{tikzpicture}
\caption{Sign epistasis}
\end{subfigure} \begin{subfigure}[b]{0.32\textwidth}
\centering
\begin{tikzpicture}[->,>=stealth',shorten >=1pt,auto,node distance=1.8cm,
                    semithick]
  \tikzstyle{every state}=[fill=white,draw=white,text=black]

  \node[state](ab){$ab$};
  \node[state](Ab)[above left of=ab]{$Ab$};
  \node[state](aB)[above right of=ab]{$aB$};
  \node[state](AB)[above right of=Ab]{$AB$};

  \path 
        (Ab) edge node {} (AB)
        	edge node {} (ab)
        (aB) edge node {} (AB)
        	edge node {} (ab)
            ;
\end{tikzpicture}
\caption{Reciprocal sign epistasis}
\label{fig:rse}
\end{subfigure}
\caption{Three different kinds of epistasis possible in fitness graphs.}
\label{fig:allepistasis}
\end{figure}

A system has \emph{sign epistasis} if it violates one of the two conditions for magnitude epistasis. For example, if $f(AB) > f(aB) > f(ab) > f(Ab)$ then there is sign epistasis at the first locus. If the second locus is b then the mutation from a to A is not adaptive, but if the second locus is B then the mutation from a to A is adaptive. 

\begin{prop}[\cite{WWC05,CGB13}]
If there is no sign epistasis in a fitness landscape then it is called a \emph{smooth landscape} or \emph{Mt. Fuji landscape} and has a single peak $x^*$. Every shortest path (ignoring edge directions) from an arbitrary $x$ to $x^*$ is an adaptive path, and vice-versa.
\label{prop:smooth}
\end{prop}

An example of an empirical smooth landscape can be seen in figure~\ref{fig:empsmooth}.

\begin{cor}
Evolution can quickly find the global optimum in a \emph{smooth landscape} on $\{0,1\}^n$ with any adaptive path taking at most $n$ steps.
\label{cor:fast}
\end{cor}

\begin{figure}
\centering
\begin{tikzpicture}[->,>=stealth',shorten >=1pt,auto,node distance=1.8cm,
                    semithick]
  \tikzstyle{every state}=[fill=white,draw=white,text=black]

  \node[state](abc){$abc$};
  \node[state](Abc)[above left of=abc]{$Abc$};
  \node[state](aBc)[above right of=abc]{$aBc$};
  \node[state](ABc)[above right of=Abc]{$ABc$};
  \node[state](AbC)[above right of=aBc]{$AbC$};
  \node[state](abC)[below right of=AbC]{$abC$};
  \node[state](aBC)[above right of=abC]{$aBC$};
  \node[state](ABC)[above right of=AbC]{$ABC$};

  \path 
  		(abc) edge[color=green] node {} (abC)
        (Abc) edge[color=red] node {} (ABc)
        	edge[color=red] node {} (abc)
            edge node {} (AbC)
        (aBc) edge[color=red] node {} (ABc)
        	edge[color=red] node {} (abc)
            edge node {} (aBC)
        (abC) edge              node {} (AbC)
            edge[color=green] node {} (aBC)
        (AbC) edge node {} (ABC)
        (aBC) edge[color=green] node {} (ABC)
        (ABC) edge[color=green] node {} (ABc)
            ;
\end{tikzpicture}
\caption{Note that the shortest paths from abc to ABc are blocked by the reciprocal sign epistasis of the red edges. However, an alternative adaptive path exists along the green edges that first introduces the C allele to reach ABC, but then removes it to return to ABc.}
\label{fig:gettingAround}
\end{figure}
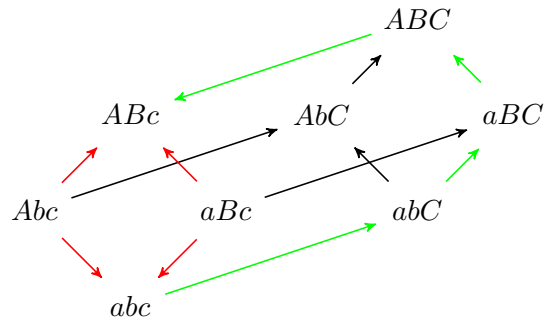

In contrast, ruggedness -- defined as having multiple peaks -- is a global property that is only weakly related to epistasis. An example of an empirical rugged landscape is given in figure~\ref{fig:emprough}. A system has \emph{reciprocal sign epistasis} if both conditions of magnitude epistasis are broken, or (equivalently) if we have sign epistasis on both loci~\cite{PKWT07}. An example of negative reciprocal epistasis would be if $f(AB) > f(ab)$ but $f(ab) > f(Ab)$ and $f(ab) > f(aB)$. The presence of reciprocal sign epistasis is a necessary condition for multiple peaks~\cite{PKWT07,PSKT11}. However, it is not sufficient for multiple peaks, since evolution can use a third locus to go around the fitness valley as shown in figure~\ref{fig:gettingAround}. In fact, there is no local property in terms of just reciprocal sign epistasis that is sufficient for the existence of multiple-peaks~\cite{CGB13}.\footnote{A sufficient condition can be given in terms of reciprocal and single sign epistasis: if there is reciprocal sign epistasis but no pair of loci with just a single sign epistasis (i.e. sign epistasis on only one of the two loci, as given in the example of the previous paragraph)~\cite{CGB13}} An open problem is if there is any local property or polynomial time testable property that is both necessary and sufficient for multiple-peaks.

\subsection{$NK$ model of rugged fitness landscapes}
\label{sec:NK}

\begin{myDef}[\cite{KW87,KW89,K93}]
	The \emph{$NK$ model} is a fitness landscape on $\{0,1\}^n$. The $n$ loci are arranged in a gene-interaction network where each locus $x_i$ is linked to $K$ other loci $x^i_1,...,x^i_k$ and has an associated fitness contribution function $f_i: \{0,1\}^{K + 1} \rightarrow \mathbb{R_+}$
	Given a vertex $v \in \{0,1\}^n$, we define the fitness $f(x) = \sum_{i = 1}^n f_i(x_ix^i_1...x^i_k)$.
\end{myDef}

By varying $K$ we can control the amount of epistasis in the landscape. With $K = 0$ we have a smooth landscape, and for higher $K$ we can treat the fitness contributions as generalizations of the two loci epistasis described in the previous section. The model also provides an upper bound of $n{K + 1 \choose 2}$ on the number of gene pairs that have epistatic interactions.

The $NK$ model is frequently studied through simulation, or statistical mechanics approaches. In a typical biological treatment, the gene-interaction network is assumed to be something simple like a generalized cycle (where $x_i$ is linked to $x_{i + 1},...x_{i + K}$) or a random $K$-regular graph. The fitness contributions $f_i$ are usually sampled from some choice of distribution. As such, we can think of biologists as doing average case analysis of the fitness landscapes. Since there is -- as we will see in section~\ref{sec:empirical} -- no empirical or theoretically sound justification for the choice of distributions, I eliminate them and focus on worst-case analysis. 

Weinberger~\cite{W96} showed that checking if the global optimum in an $NK$ model is greater than some input value $V$ is $NP$-complete for $K \geq 3$. Although this implies that finding a global optimum is difficult, it says nothing about local optima. As such, it has generated little interest among biologists, although it spurred interest as a model in the evolutionary algorithms literature, leading to a refined proof of $NP$-completeness for $K \geq 2$~\cite{WTZ00}.

\subsection{Empirical fitness landscapes and speed of adaptation}
\label{sec:empirical}

\begin{figure}
\centering
\begin{subfigure}[b]{0.49\textwidth}
\centering
\begin{tikzpicture}[->,>=stealth',shorten >=1pt,auto,node distance=2.8cm,
                    semithick]
  \tikzstyle{every state}=[fill=white,draw=white,text=black,scale=0.5]

  \node[state](0000)[text=blue]{$0000$};
  \node[state](0100)[above left of=0000,text=blue]{$0100$};
  \node[state](0010)[above right of=0000,text=blue]{$0010$};
  \node[state](1000)[left of=0100,text=blue]{$1000$};
  \node[state](0001)[right of=0010,text=blue]{$0001$};
  \node[state](1100)[above left of=1000,text=blue]{$1100$};
  \node[state](1010)[right of=1100,text=blue]{$1010$};
  \node[state](0110)[right of=1010,text=blue]{$0110$};
  \node[state](0011)[above right of=0001,text=blue]{$0011$};
  \node[state](0101)[left of=0011,text=blue]{$0101$};
  \node[state](1001)[left of=0101,text=blue]{$1001$};
  \node[state](1110)[above right of=1100,text=blue]{$1110$};
  \node[state](1101)[right of=1110,text=blue]{$1101$};
  \node[state](0111)[above left of=0011,text=blue]{$0111$};
  \node[state](1011)[left of=0111,text=blue]{$1011$};
  \node[state](1111)[above right of=1101,text=blue]{$\underline{1111}$};
  
  \path
  (0000) edge node {} (1000)
  	edge node {} (0100)
    edge node {} (0010)
    edge node {} (0001)
  (1000) edge node {} (1100)
  	edge node {} (1010)
    edge node {} (1001)
  (0100) edge node {} (1100)
  	edge node {} (0110)
    edge node {} (0101)
  (0010) edge node {} (1010)
  	edge node {} (0110)
    edge node {} (0011)
  (0001) edge node {} (1001)
  	edge node {} (0101)
    edge node {} (0011)
  (1100) edge node {} (1110)
  	edge node {} (1101)
  (1010) edge node {} (1110)
  	edge node {} (1011)
  (1001) edge node {} (1101)
  	edge node {} (1011)
  (0110) edge node {} (1110)
  	edge node {} (0111)
  (0101) edge node {} (1101)
  	edge node {} (0111)
  (0011) edge node {} (1011)
  	edge node {} (0111)
  (1110) edge node {} (1111)
  (1011) edge node {} (1111)
  (1101) edge node {} (1111)
  (0111) edge node {} (1111)
	;

\end{tikzpicture}
\caption{\emph{Escherichia coli} $\beta$-lactamase}
\label{fig:empsmooth}
\end{subfigure} \begin{subfigure}[b]{0.49\textwidth}
\centering
\begin{tikzpicture}[->,>=stealth',shorten >=1pt,auto,node distance=2.8cm,
                    semithick]
  \tikzstyle{every state}=[fill=white,draw=white,text=black,scale=0.5]

  \node[state](0000)[text=green]{$0000$};
  \node[state](0100)[above left of=0000,text=green]{$0100$};
  \node[state](0010)[above right of=0000,text=green]{$0010$};
  \node[state](1000)[left of=0100,text=blue]{$1000$};
  \node[state](0001)[right of=0010,text=red]{$0001$};
  \node[state](1100)[above left of=1000,text=blue]{$\underline{1100}$};
  \node[state](1010)[right of=1100,text=blue]{$1010$};
  \node[state](0110)[right of=1010,text=green]{$0110$};
  \node[state](0011)[above right of=0001,text=red]{$0011$};
  \node[state](0101)[left of=0011,text=red]{$\underline{0101}$};
  \node[state](1001)[left of=0101,text=green]{$1001$};
  \node[state](1110)[above right of=1100,text=blue]{$1110$};
  \node[state](1101)[right of=1110,text=green]{$1101$};
  \node[state](0111)[above left of=0011,text=red]{$0111$};
  \node[state](1011)[left of=0111,text=green]{$1011$};
  \node[state](1111)[above right of=1101,text=green]{$1111$};
  
  \path
  (0000) edge node {} (1000)
  	edge node {} (0100)
    edge node {} (0010)
  (1000) edge node {} (1100)
    edge node {} (1001)
  (0100) edge node {} (1100)
    edge node {} (0101)
  (0010) edge node {} (1010)
  	edge node {} (0110)
  (0001) edge node {} (1001)
  	edge node {} (0101)
    edge node {} (0000)
  (1010) edge node {} (1110)
  	edge node {} (1011)
    edge node {} (1000)
  (0110) edge node {} (1110)
  	edge node {} (0111)
    edge node {} (0100)
  (1001) edge node {} (1101)
  (0011) edge node {} (1011)
  	edge node {} (0111)
    edge node {} (0100)
    edge node {} (0010)
    edge node {} (0001)
  (1110) 
  	edge node {} (1100)
  (1011) edge node {} (1111)
  	edge node {} (1001)
  (1101) 
  	edge node {} (1100)
    edge node {} (0101)
  (0111) 
  	edge node {} (0101)
  (1111) edge node {} (0111)
  	edge node {} (1101)
    edge node {} (1110)
	;
\end{tikzpicture}
\caption{\emph{Plasmodium falciparum} dihydrofolate reductase}
\label{fig:emprough}
\end{subfigure}
\caption{Two examples of empirical fitness landscape from Figure 1 of~\cite{SSFKdV13}. Figure~\ref{fig:empsmooth} is based on the data of Chou et al.~\cite{Cetal11} and contains a single optimum (1111) and is a smooth landscape with no sign epistasis. Figure~\ref{fig:emprough} is based on data from Lozovsky et al.~\cite{Letal09} and has both single sign and reciprocal sign epistasis and two peaks. The first peak is 1100 and its basin of attraction is shown in blue, the second is 0101 with a red basin of attraction; vertices in green have equal length shortest path to each fitness peak.}
\label{fig:empirical}
\end{figure}
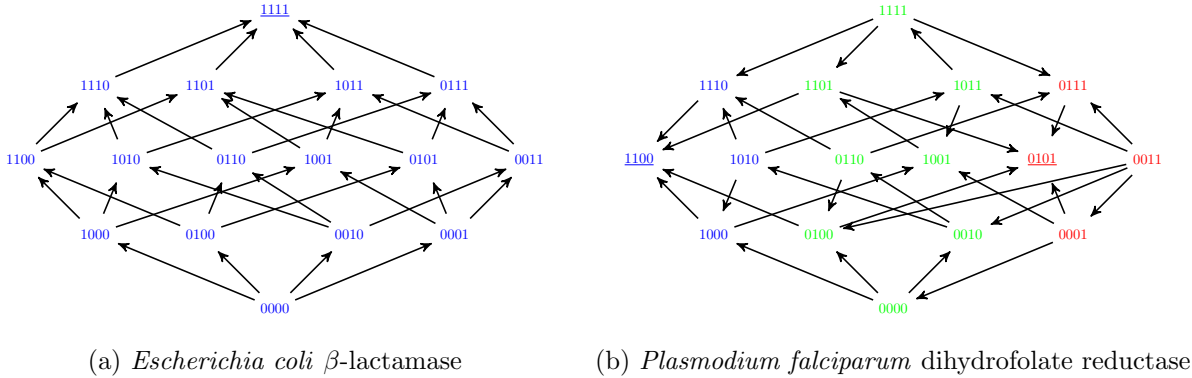

No local property is known for ruggedness, and this makes empirical tests extremely difficult~\cite{WPMT95,KTP09}. In particular, most experimental results do not measure the fitness landscape, but only report the average fitness versus time and average number of acquired adaptations versus time~\cite{LT94,CL00,Betal09,KTP09}. Szendro et al.~\cite{SSFKdV13} surveyed the few recent experiments that conducted a methodical examination of all mutations in a subset of loci of model organisms, but most studies (6 out of 12; two are presented in figure~\ref{fig:empirical}) were able to empirically realize only small fitness landscapes of just 4 to 5 loci, with the largest full fitness landscape having length 6~\cite{HAP10}, and the largest number of vertices in a single study being 418 out of the possible 512 in a length 9 landscape~\cite{OMetal08}. 
Unfortunately, a four loci landscape is simply too local of a property, and not much more informative than the reductionist two loci analysis of epistasis. However, the biological intuition is that real landscapes are a little rough, and have multiple optima but not as frequent as completely uncorrelated models. Although, we know almost nothing about the structure of fitness landscapes, theoretical biologists continue to forge ahead, making arbitrary assumptions about the structure and statistical properties of their models. There is a disconnect between theory and data~\cite{O05,KTP09}.

Is the assumption of local equilibrium empirically reasonable? From a genome-wide perspective, it seems to be at odds with the intuition of naturalists. Consider vestigial features of your own body like your appendix, goose bumps, tonsils, wisdom teeth, third eyelid, or the second joint in the middle of your foot made immobile by a tightened ligament. Would it not be more efficient (and thus produce marginally higher fitness) if you did not spend the energy to construct these features? Of course, this naturalist argument is not convincing since we don not know if there are any small mutations that could remove these vestigial features from our development, I could just be describing a different local optimum that lays on the other side of a fitness valley from my current vertex.\footnote{This example is further complicated because the concept of equilibrium is different for sexual organisms that are capable of recombination and often does not correspond to something as simple as a peak in the fitness landscape~\cite{LPDF08}}. The other tempting naturalist example is macroevolutionary change like speciation. Unfortunately, on such long timescales the environment is not constant and depends on the extant organisms through mechanisms like niche-construction or frequency-dependence.\footnote{This defense of evolutionary equilibria is a central part of the punctuated equilibrium theory of evolution; the environment changes and the wild-type becomes not locally optimal. Adaptation is assumed to quickly carry the species to a nearby local optimum where it remains for a long period of time until the next environmental change.}

One of the earliest successes for mathematical models of rugged landscapes was an application to affinity maturation in immune response~\cite{KW89}. The length of evolutionary process leading to affinity maturation is very short, typically adaptation stops after only 6-8 nucleotide changes~\cite{CGHCH81,T83,Cetal85} -- an adaptive process that happens on the order of days. However, the results should be taken with a grain of salt, the adapted B-cells were not experimentally isolated and all of their point-mutations were not checked to guarantee that a fitness peak was reached. In both theoretical and experimental treatments of evolution, it is known that fitness increases tend to show a pattern of geometrically diminishing returns~\cite{LT94,O98,CL00,KTP09} which means that after a few generations the fitness change will be so small that the fixation time will be longer than the presence of the pathogen causing the immune response. Further, the activation-induced (DNA-cytosine) deaminase enzyme (and other mechanisms) increases the rate of mutation by a factor of $10^6$ along the gene encoding antibody proteins, suggesting that this is a fitness landscape that has been shaped by evolution of the human immune system to find fit mutants as quickly as possible. This creates a bias toward landscapes where local maxima would be easier to find than usual, and thus makes affinity maturation a poor candidate for considering the speed of typical evolution.

A more typical setting might be the evolution of \emph{E. coli} in a static fitness landscape. Here, biologists have run long-term experiments tracking a population for over 50,000 generations~\cite{LT94,CL00,BBDL12} and continue to find adaptations and marginal increases in fitness. This suggests that a local optimum is not quickly found, even though the environment is static. However, it is difficult to estimate the number of adaptive mutations that fixed in this population, and Lenski~\cite{L03} estimated that as few as 100 adaptive point-mutations fixated in the first 20,000 generations. It is also hard to argue that the population doesn't traverse small fitness valleys between measurements, which could be used to suggest that the colony is hopping from one easy-to-find local equilibrium to the next.

\section{$NK$ model with $K \geq 2$ is PLS-complete}
\label{sec:PLS}

To understand the difficulty of finding items with some local property like being an equilibrium, Johnson, Papadimitrio \& Yannakakis~\cite{PLS} defined the complexity class of polynomial local search (PLS). A problem is in PLS if it can be specified by three polynomial time algorithms~\cite{R09}:

\begin{enumerate}
\item An algorithm $I$ that accepts an instance (like a description of a fitness landscape) and outputs a first candidate to consider (the wild type).
\item An algorithm $F$ that accepts an instance and a candidate and returns a objective function value (computes the fitness).
\item An algorithm $M$ that accepts an instance and a candidate and returns an output with a strictly higher objective function value, or says that the candidate is a local maximum.
\end{enumerate}

We consider a PLS problem solved if an algorithm can output a locally optimal solution for every instance.\footnote{This algorithm does not necessarily have to use $I$, $F$, or $M$ or follow adaptive paths. For instance, it can try to uncover hidden structure from the description of the landscape. A classical example would be the ellipsoid method for linear programming.} The hardest problems in PLS -- i.e. ones for which a polynomial time solution could be converted to a solution for any other PLS problem -- are called PLS-complete. It is believed that PLS-complete problems are not solvable in polynomial time, but -- much like the famous $\mathsf{P} \neq \mathsf{NP}$ question -- this conjecture remains open. Note that finding local optima on fitness landscapes is an example of a PLS problem, where $I$ is your method for choosing the wild type, $F$ is the fitness function, and $M$ is an adaptive step.

\begin{myDef}[Weighted 2SAT]
Consider $n$ variables $x = x_1...x_n \in \{0,1\}^n$ and $m$ clauses $C_1,...,C_m$ and associated positive integer weights $c_1,...c_m$. Each clause $C_k$ contains two literals (a literal is a variable $x_i$ or its negation $\bar{x_i}$), and contributes $c_k$ to the fitness if at least one of the literals is satisfied, and nothing if neither literal is satisfied. The total fitness $c(x)$ is the sum of the individual contributions of the $m$ clauses. Two instances $x$ and $x'$ are adjacent if there is exactly one index $i$ such that $x_i = x'_i$. We want to maximize fitness.
\end{myDef}

The Weighted 2SAT problem is PLS-complete~\cite{SY91}. To show that the $NK$ model is also PLS-complete, I will show how to reduce any instance of Weighted 2SAT to an instance of the $NK$ model.

\begin{thm}
	Finding a local optimum in the NK fitness landscape is $PLS$-complete.
	\label{thm:main}
\end{thm}

For $K \leq 1$, even a global optimum can be found in polynomial time~\cite{WTZ00}, so the theorem is as strong as it can be.

\begin{proof}
	Consider an instance of Weighted 2SAT with variables $x_1,...,x_n$, clauses $C_1,...,C_m$ and positive integer costs $c_1,...,c_m$. We will build a landscape with $m + n$ loci, with the first $m$ labeled $b_1,...,b_m$ and the next $n$ labeled $x_1,...,x_m$. Each $b_k$ will correspond to a clause $C_k$ that uses the variables $x_i$ and $x_j$ (i.e., the first literal is either $x_i$ or $\bar{x_i}$ and the second is $x_j$ or $\bar{x_j}$; set $i < j$ to avoid ambiguity). Define the corresponding fitness effect of the locus as:
	\begin{align}
		f_{k}(0x_ix_j) & =  \begin{cases}
		c_k & \mathrm{if}\;C_k\;\mathrm{is}\;\mathrm{satisfied} \\
		0 & \mathrm{otherwise}
		\end{cases} \\
		f_{k}(1x_ix_j) & = f_{k}(0x_ix_j) + 1
	\end{align}
	Link the $x_i$ arbitrarily (say to $x_{(i \mod n) + 1}$ and $x_{(i + 1 \mod n) + 1}$, or to nothing at all) with a fitness effect of zero, regardless of the values.
	
	In any local maximum $b,x$, we have $b = 11..1$ and $f(x) = m + c(x)$. On the subcube with $b = 11..1$  Weighted 2SAT and this NK model have the same exact fitness graph structure, and so there is a bijection between their local maxima.
\end{proof}

Assuming -- as most computer scientists do -- that there exists some problem in PLS not solvable in polynomial time, then theorem~\ref{thm:main} implies that no matter what mechanistic rule evolution follows (even ones we have not discovered, yet), be it as simple as SSWM or as complicated as any polynomial time algorithm, there will be NK landscapes with $K = 2$ such that evolution will not be able to find a fitness peak efficiently. But if we focus only rules that follow the adaptive paths then we can strengthen the result:

\begin{cor}
	There is a constant $c > 0$ such that, for infinitely many $n$, there are instances of $NK$ models (with $K \geq 2$) on $\{0,1\}^n$ and initial vertices $v$ such that any adaptive path from $v$ will have to take at least $2^{cn}$ steps before finding a fitness peak.
\end{cor}

\begin{proof}
If the initial vertex has $s = 11...1$ then there is a bijection between adaptive paths in the fitness landscape and any weight-increasing path for optimizing the weighted 2SAT problem. Thus, theorem~5.15 of~\cite{SY91} applies.
\end{proof}

This result holds independent of the relationship between polynomial time and PLS. There are some landscapes and initial organisms, such that any rule we use for adaptation that only considers fitter single-gene mutants will take an exponential number of steps to find the local optimum.

\section{Exponential adaptive paths on non-rugged landscapes}
\label{sec:simplex}

A fitness landscape where there is no short adaptive path to a fitness peak is certainly very complicated. However, even in landscapes with short adaptive paths, it is not always possible for evolution to find them. This follows from an algorithm analysis result~\cite{MS06} showing that the random edge simplex algorithm can be exponential on abstract cubes. In particular, I show that there exist single-peaked fitness landscapes with no reciprocal sign epistasis, on which the expected length of an adaptive path following SSWM to the optimum from a randomly selected vertex will (with high probability) be super-polynomial. It is not surprising to find the simplex algorithm in this context, since we can regard it as a local search algorithm for linear programming where local optimality coincides with global optimality.

\begin{myDef}
A directed acyclic orientation of a hypercube $\{0,1\}^n$ is called an \emph{acyclic unique sink orientation (AUSO)} if every subcube (face; including the whole cube) has a unique sink.
\end{myDef}

A fitness graph that is an AUSO has a single peak and is not rugged. Since, a complete lack of epistasis produces a smooth landscape by prop.~\ref{prop:smooth} and by cor.~\ref{cor:fast} evolution is efficient. Introducing sign epistasis is the smallest modification we can do in terms of epistatic structure to produce a landscape in which evolution is not efficient. Since every subcube of an AUSO has a unique sink, it means that any empirical observation of a few loci of a bigger AUSO will result in an empirical AUSO. In particular, this means that we already know empirical landscapes that are more complicated than AUSOs, for example the fitness graph in figure~\ref{fig:emprough}.

\begin{prop}
Fitness landscapes without reciprocal sign epistasis are AUSOs
\end{prop}

\begin{proof}
($\Leftarrow$) In an AUSO, a subcube like the one in figure~\ref{fig:rse} cannot exist because it has two sinks, therefore an AUSO has no reciprocal sign epistasis.
($\Rightarrow$) Take any subcube of $\{0,1\}^n$, since we only removed vertices from the graph, we cannot have introduced any reciprocal sign epistasis where there was not before ,so by prop~\ref{prop:smooth} there is a unique fitness maximum (sink).
\label{prop:AUSO}
\end{proof}

\begin{prop}
	If there is no reciprocal sign epistasis then there always exists an adaptive path from any vertex $v$ to the unique fitness maximum $v^*$ of length equal to the number of bits on which $v$ and $v^*$ differ.
	\label{prop:short}
\end{prop}

\begin{proof}
	This proposition is true for $\{0,1\}^1$; continue by induction:
	\begin{enumerate}
	\item If $u$ is at distance $n$ from $v^*$ then all adjacent vertices are at lower distance, and since $u$ is not a sink there must be an adaptive edge from it to vertex $u'$ by flipping $u_i$. Look at the subcube with the $i$-th bit fixed to $u_i$, by IH $u'$ has an adaptive path of length $n - 1$ to $v^*$.
	\item If $u$ is at distance $k < n$ from $v^*$. Look at the $k$-subcube of variables on which $u$ and $v^*$ differ (i.e. our subcube consists of all vertex $v'$ such that if $u_i = v^*_i$ then $u_i = v'_i$), by IH $u$ has an adaptive path of length $k$ to $v^*$ in this subcube.
	\end{enumerate}
\end{proof}

Proposition~\ref{prop:short} and the minimal change from a smooth landscape might suggest that the optimum in an AUSO will be easy to find, but a result adapted from the analysis of simplex algorithms shows this intuition to be misguided:

\begin{thm}[\cite{MS06} in biological terminology]
There are positive constants $c_1,c_2$ such that for all sufficiently large $n$ there exists an AUSO fitness landscape on $\{0,1\}^n$ such that SSWM dynamics starting from a random vertex, with probability at least $1 - e^{-c_1n^{1/3}}$ follows an adaptive path of at least $e^{c_2n^{1/3}}$ steps to the fitness maximum.
\end{thm}

In other words, multiple peaks or even reciprocal sign-epistasis are not required to make a complex fitness landscape. In fact, AUSOs were developed to capture the idea of a linear function on a polytope. Linear fitness functions are usually considered to be some of the simplest landscapes by theoretical biologists; showing that adaptation is hard on these landscapes or ones like them is a surprising result. Although most AUSOs are not given by any linear function (including the one constructed by Matousek \& Szabo~\cite{MS06}), they are `close in spirit' to linear. Further, similar lower bounds have been shown for linear functions~\cite{FHZ11}, but unfortunately the polytope is not a hypercube. An open question at the intersection of the analysis of simplex algorithms and fitness landscapes is to find a linear function on hypercubes for which SSWM dynamics cannot efficiently find the optimum.

\section{Discussion}

I have shown that the assumption of an easily reachable fitness peak is not reasonable in the $NK$ model with $K \geq 2$, or even in some single-peaked landscapes without any reciprocal sign epistasis. Even with an adaptive path of length equal to the number of differing bits existing, it is possible to have landscapes where SSWM dynamics take an exponential number of steps. In the $NK$ model, going from epistasis restricted to pairs to triplets (i.e. from $K = 1$ to $K =  2$) makes the model as complex as any fitness landscape; and makes local fitness peaks impossible to find in polynomial time (under a standard computational complexity assumption) regardless of what process evolution uses. Since a static fitness landscape is only an accurate approximation on short to moderate timescales, we cannot expect long adaptive walks to have a chance to converge. This means one of three things for biology:

\begin{enumerate}
\item Abandon the fitness landscape metaphor.
\item Redefine or restrict existing fitness landscape models, and show that these modifications are empirically reasonable and that fitness peaks are efficiently reachable in these models. Try to provide a general treatment with as few assumptions as possible.
\item Accept that local equilibrium assumptions are unjustified, and develop a theory that is comfortable with organisms that are far from equilibrium and always have more room to adapt.
\end{enumerate}

The first option has already been pursued by theoretical computer scientists in earlier work  on evolution. Valiant's model~\cite{evolvability} moves past the fitness landscape metaphor by suggesting that we view organisms as protein circuits and evolution as a learning algorithm that is trying to have the organisms approximate an ideal function. This approach provides a beautiful synthesis of machine learning, artificial intelligence, and evolution, but is not presented in a language familiar to biologists or form that is easily ammendable to experimental studies. This subfield uses significantly more powerful techniques than I used here, and computer scientists have steadily moved Valiant's model forward to address difficulties that I did not deal with explicitly in this presentation; like drifting evolutionary targets~\cite{KVV10}, and the effects of sexual reproduction with recombination~\cite{K11}. Unfortunately, the last five years have shown that biologists do not favour the abandonment of fitness landscapes, with citations coming only from other computer scientists.

The second option has seen a healthy development in Orr-Gillespie theory~\cite{G91,O02,O06}. This approach is based on Gillespie's insight that the wild-type (initial vertex) represents a draw from the tail of a fitness distribution, and beneficial mutations are even more extreme draws from this tail. At each step, evolution samples several times from the distribution and if one of the mutants has higher fitness than our current value, we take an adaptive step, and repeat the process. By using extreme value theory, Orr-Gillespie makes fewer assumptions about underlying distributions than previous treatments and only has to specify a tail shape parameter (although this is still more than the worst-case analysis used here). Unfortunately, sampling from fixed fitness distributions (or even ones that depend on the current fitness, but not position in the landscape~\cite{KTP09}) is an extremely local approximation, and ignores the combinatorial structure of the landscape. Without looking at the graph of adjacent genotypes, it is simply not possible to study epistasis fully, even if some single adaptive path measures have been suggested~\cite{KTP09}.

The third option is my preference because I believe it is the most productive. There are many qualities, like biological complexity~\cite{M91}, that seem to increase constantly during evolution. Since it is natural to expect that the complexity fitness landscape will be complex, my results on the unreachability of local peaks provides an explanation for the constant growth. Since the adaptive walk will take an exponential number of steps to converge, we will never see it reach equilibrium and thus always observe an increase in complexity. Without the assumption of being at local equilibrium, we also have to change our language from ``adapted for" to ``adapting to", which helps avoid teleological tangents. If we are not allowed to assume that organisms are at local fitness peaks, then we have to always think of evolution as a process, not a destination.

\section*{Acknowledgements}

I would like to thank Julian Z. Xue for introducing me to the $NK$ model, and for many helpful discussion; and David Basanta and Jacob Scott for hosting me at the Moffitt Cancer Research Institute where I could present the preliminaries of this work to an audience of mathematical biologists. This paper benefited from comments by Eric Bolo, Daniel Nichol, and Prakash Panangaden.

\bibliography{references}
\bibliographystyle{alpha}

\end{document}